\documentclass{article} 
\usepackage{authblk}
\usepackage{hyperref}

\title{Exact Cut Complexity of Equal-Length Proportional Cake Cutting\thanks{A preliminary version of this paper appeared in the proceedings of the 51st International Symposium on Mathematical Foundations of Computer Science (MFCS 2026) \cite{kawase_sanpui_2026}.\\
E-mail addresses: {ykawase422@g.chuo-u.ac.jp} (Y. Kawase), bodhayan.roy@gmail.com (B. Roy), azharuddinsanpui123@gmail.com (M. A. Sanpui).   }}
\author[1]{Yasushi Kawase}
\author[2]{Bodhayan Roy}
\author[2]{Mohammad Azharuddin Sanpui}
\affil[1]{Chuo University,\authorcr Tokyo, Japan \authorcr \mbox{}}
\affil[2]{  Indian Institute of Technology Kharagpur,\authorcr
    Kharagpur, West Bengal,
            India}
            
\date{}

\usepackage{fullpage}
\usepackage[linesnumbered,ruled,vlined]{algorithm2e}

\SetCommentSty{mycommfont}

\usepackage{amsthm,amsmath,amssymb,mathtools}
\newtheorem{theorem}{Theorem}
\newtheorem{definition}{Definition} 
\newtheorem{proposition}{Proposition}

\newtheorem{lemma}{Lemma}
\newtheorem{corollary}{Corollary}
\usepackage{hyperref}
\usepackage[capitalise]{cleveref}
\theoremstyle{remark}
\newtheorem{remark}[theorem]{Remark}         

\crefname{hypothesis}{Hypothesis}{Hypotheses}
\crefname{fact}{Fact}{Facts}
\crefname{claim}{Claim}{Claims}
\usepackage{colortbl}
\newcolumntype{d}{>{\columncolor{gray!30}}c}
\usepackage{booktabs}

\usepackage[numbers,sort]{natbib} 
\usepackage{cleveref}
\usepackage{enumitem}
\setlist[itemize]{leftmargin=10pt}

\newcommand{\C}{[0,1]}
\newcommand{\length}[1]{\lambda\left(#1\right)}
\newcommand{\kprop}[1]{\kappa_{#1}^{\mathrm{prop}}}
\newcommand{\kef}[1]{\kappa_{#1}^{\mathrm{ef}}}
\newcommand{\kefh}[2]{\kappa_{#1,#2}^{\mathrm{ef}}}

\usepackage{tikz}
\usetikzlibrary{positioning,shapes.misc,calc}

\usepackage{todonotes}
\definecolor{myyellow}{cmyk}{0,0.02,0.23,0.01}

\begin{document}

\maketitle

\begin{abstract}
We study proportional cake cutting under an equal-length constraint, where each of $n$ agents must receive a bundle of length exactly $1/n$ that is worth at least $1/n$ to that agent.
Our main result determines the exact worst-case number of cuts needed to guarantee such an allocation, namely $2n-2$ for every $n$. This improves on the previously known quadratic upper bound.
The proof combines the constrained necklace-splitting theorem with a recursive assignment argument based on Hall's theorem.
The same bound extends to the stronger requirement of envy-freeness when $n$ is a prime power, and an $O(n^{1.525})$ upper bound holds for arbitrary $n$.
We then ask whether such allocations can be computed.
We first show that no deterministic protocol in the Robertson--Webb model terminates after finitely many queries on every instance and always returns an exact equal-length proportional allocation, even for two agents with identical valuations.
For every $\varepsilon>0$, however, $O(n^2/\varepsilon)$ queries suffice to produce an allocation in which every agent receives a bundle of length exactly $1/n$ and value at least $1/n-\varepsilon$, again using at most $2n-2$ cuts.
Finally, we augment the Robertson--Webb model with two-knife queries.
We show that recursive balanced halving attains the optimal bound of $2n-2$ cuts when $n$ is a power of two.
For arbitrary $n$, we give a procedure using $O(n\log n)$ cuts.
\end{abstract}


\section{Introduction}\label{sec:introduction}
Fair division asks how to allocate divisible and indivisible resources among agents with heterogeneous preferences \cite{Steinhaus1949SurLD,Amanatidis2023FairDO,Bouveret2016FairAO}.
Cake cutting is its standard model for divisible resources.
Proportional and envy-free cake cutting have been studied extensively, with existential, algorithmic, and complexity results for many variants \cite{cseh2020complexity,Bei2012OptimalPC,aziz2016discrete,edmonds_pruhs_2011,Pikhurko2000OnEC,Su1999RentalHS}.

In many fair-division settings, the size of each share is fixed independently of the agents' preferences.
Appointment slots have prescribed durations, bandwidth blocks have fixed widths, and land plots may have standardized areas.
In each case, fairness must be achieved subject to an exogenous size constraint.

We consider the interval cake $[0,1]$, where each of the $n$ agents must receive a bundle of length exactly $1/n$.
In the terminology of fair division with market values \cite{barman_ebadian_latifian_shah_2025}, Lebesgue length is a common objective valuation.
We treat this measure as public information and the agents' valuations as private.
Proportionality is imposed with respect to the private valuations.
Thus, each agent's bundle must have private value at least $1/n$ and length exactly $1/n$.

Existence alone does not control fragmentation.
A schedule split among many short intervals can be inconvenient, and a land allotment divided into many patches can be costly to use.
On an interval, the natural measure of fragmentation is the number of cuts.
We ask for the minimum number of cuts that always suffices to give every agent a bundle of length $1/n$ and value at least $1/n$.
Without the equal-length constraint, the classical proportionality guarantee can be attained with connected pieces and $n-1$ cuts \cite{Steinhaus1949SurLD,segal_halevi_2019}.
Our main result shows that imposing equal lengths increases the exact worst-case cut complexity to $2n-2$.
Throughout this paper, a cut is a distinct boundary point retained in the final allocation and measures fragmentation.
It is not a Robertson--Webb \textsc{Cut} query, which is an information-acquisition operation counted separately in \Cref{sec:uncomputable}.

Before this work, the only upper bound for all $n$ was quadratic.
It follows from the corresponding envy-free problem with market values \cite{barman_ebadian_latifian_shah_2025}, which imposes the same constraint on the common objective measure as our setting but strengthens proportionality to envy-freeness for the private valuations.

\subsection{Our Results}\label{sec:our-results}

Let $\kprop{n}$ denote the least number of cuts that always suffice for equal-length proportional cake cutting with $n$ agents.
Our main theorem determines this number exactly
$\kprop{n}=2n-2$ for all $n\ge 1$.
An identical-valuation instance gives the lower bound (\Cref{prop:elementary-lower-bound}). 
The remainder of \Cref{sec:prop-cut} proves the matching upper bound.

To show this, we use the constrained necklace-splitting theorem of Joji\'c, Panina, and \v{Z}ivaljevi\'c \cite{jojic_panina_zivaljevic_2021}.
For a prime-power divisor $r$ of $n$, it produces $r$ equal-length top-level bundles while balancing a fractional assignment of agents to bundles.
Hall's theorem converts the fractional support into $r$ equally sized agent groups, after which we recurse within each bundle.
This gives the recurrence in \Cref{lem:prop-rec} and the exact bound in \Cref{thm:main}.

We also consider the stronger requirement of equal-length envy-freeness.
Let $\kef{n}$ denote its worst-case cut complexity.
When $n$ is a prime power, the constrained necklace-splitting theorem and the proportional lower bound give $\kef{n}=2n-2$ (\Cref{prop:prime-power-direct}).
For arbitrary $n$, a peeling argument based on the Stromquist--Woodall exact-share theorem gives $\kef{n}=O(n^{1.525})$ (\Cref{prop:ef-bhp}).
Determining $\kef{n}$ exactly for general $n$ remains open.

Existence does not imply finite query complexity.
Even with free access to length, no deterministic protocol in the Robertson--Webb model terminates after finitely many queries on every instance and always returns an exact equal-length proportional allocation.
This impossibility holds even for two agents with the same valuation (\Cref{thm:main-un}).
For every $\varepsilon>0$, however, discretization computes an equal-length $\varepsilon$-proportional allocation using $O(n^2/\varepsilon)$ queries and at most $2n-2$ cuts (\Cref{thm:rw-approx}).
A polynomial-time alternative in the real-RAM oracle model uses $O(n^3/\varepsilon)$ queries and $O(n^2/\varepsilon)$ cuts (\Cref{prop:rw-approx-polytime}).

With two-knife queries in addition to Robertson--Webb queries, equal-length proportional allocations can be produced without approximation.
When $n$ is a power of two, recursive balanced halving uses \textsc{MoveLength} queries, which move a fixed-length window until a prescribed number of agents value it at least their specified thresholds.
The procedure uses at most $n-1$ two-knife queries and the optimal $2n-2$ cuts (\Cref{thm:moving-power-two}).
For arbitrary $n$, a \textsc{MoveValue} query follows Austin's two-knife procedure by keeping one agent's value between the knives fixed and stopping when the portion between them reaches a specified length.
Combining this query with a matching-based recursion gives an equal-length proportional allocation using $O(n^5)$ Robertson--Webb queries, $O(n^2)$ two-knife queries, and $O(n\log n)$ cuts for arbitrary nonatomic valuations (\Cref{thm:rw-upper}).

\subsection{Related Work}\label{sec:related-work}

The most directly related work concerns fair division with market values.
Barman, Ebadian, Latifian, and Shah \cite{barman_ebadian_latifian_shah_2025} study fair division with both private valuations and a common market valuation.
Taking that common valuation to be Lebesgue length gives the constraint considered here.
In the cake-cutting setting, their result gives the upper bound $\kef{n}\le n(n-1)$ for the stronger requirement of envy-freeness with respect to both the private valuations and the common market valuation, and their Open Question~5 asks whether $O(n)$ cuts always suffice.
We do not settle that question.
In the equal-length setting, \Cref{sec:ef-cut} improves their quadratic bound to $O(n^{1.525})$.
Since envy-freeness implies proportionality, their quadratic bound also applied to the proportional problem, but it did not determine its exact cut complexity.

Equal-size constraints also appear for indivisible goods.
Biswas and Barman study cardinality constraints \cite{biswas_barman_2018}, while Kawase and Mahara study balanced allocations, in which all agents receive the same number of items \cite{kawase_mahara_2026}.
Related models appear in \cite{bu2023fair,golz2026fair,Kawase2025SimultaneouslyFA,la2026bad}.

Proportional cake cutting with unequal or weighted entitlements also minimizes cuts subject to private-value guarantees.
In that setting, Segal-Halevi \cite{segal_halevi_2019} proved a general $O(n\log n)$ upper bound together with a $2n-2$ lower bound, and Crew, Narayanan, and Spirkl \cite{crew_narayanan_spirkl_2020} improved the upper bound to $3n-4$.
These results do not require each allocated bundle to have length exactly $1/n$, or more generally common-measure value exactly $1/n$, and therefore do not imply the present theorem.

Exact division and continuous necklace splitting require every agent to value all final bundles equally.
This stronger requirement yields quadratic cut bounds, for example via Alon \cite{alon_1987}; see also Goldberg and Li \cite{goldberg_li_itcs2023} for exact divisions into prescribed ratios.
Exact division is stronger than equal-length proportionality, which requires only the assigned bundle to be worth at least $1/n$ to each agent.
We instead use the constrained necklace-splitting theorem of Joji\'c, Panina, and \v{Z}ivaljevi\'c \cite{jojic_panina_zivaljevic_2021}.
Its preference version simultaneously equalizes length and produces an accepted assignment.
For a prime-power number of agents, choosing preferences according to private values directly gives the optimal $2n-2$ bound.

Exact fairness can be incompatible with finite Robertson--Webb interaction.
Exact division is a classical example \cite{robertson_webb_1997}.
Procaccia and Wang~\cite{procaccia2017lower} give quantitative lower bounds for approximate equitability and show that exact equitability cannot be computed even without a uniform query bound.
Here length is known and only the private valuation is queried.
The proof in \Cref{sec:uncomputable} exploits this distinction by perturbing private masses at geometric cuts and geometric positions at value cuts.
For unconstrained proportional cake cutting, Arndt, Moseley, Im, and Pruhs~\cite{arndt_moseley_im_pruhs_2026} recently gave a deterministic algorithm using $O(n)$ \textsc{Cut} queries and $O(n^2)$ \textsc{Eval} queries.
Their result does not impose equal lengths and counts query operations rather than cuts retained in the final allocation.

Austin's two-knife procedure assumes that an agent can move two knives continuously while keeping the intervening piece at a fixed value~\cite{Austin1982}.
Our additional query uses this continuous motion and stops when the intervening piece reaches a specified length.
\section{Preliminaries}\label{sec:preliminaries}

We consider allocating the cake $\C$ among $n$ agents, indexed by $\{1,\dots,n\}$.
Each agent $i$ has a valuation $v_i$, modeled as a nonatomic Borel probability measure on $\C$.
Thus, $v_i(\C)=1$.
Let $\lambda$ denote Lebesgue measure on $[0,1]$.

An allocation is a measurable partition $(C_1,\dots,C_n)$ of $\C$.
We call it \emph{equal-length} if $\length{C_i}=1/n$ for every $i\in\{1,\dots,n\}$.
We call it \emph{proportional} if $v_i(C_i)\ge 1/n$ for every $i\in\{1,\dots,n\}$, and \emph{envy-free} if $v_i(C_j)\le v_i(C_i)$ for every $i,j\in\{1,\dots,n\}$.
Every envy-free allocation is proportional because $v_i(C_i)$ is then at least the average of $v_i(C_1),\dots,v_i(C_n)$, which is $1/n$.

Given $q$ distinct interior cut points, list them as $0<t_1<\cdots<t_q<1$, put $t_0=0$ and $t_{q+1}=1$, and let $m=q+1$.
For $q=0$, this list is empty and $m=1$.
For $j\in\{1,\dots,m\}$, define the induced pieces by
\[
  I_j\coloneqq
  \begin{cases}
    [t_0,t_1] & \text{if } j=1,\\
    (t_{j-1},t_j] & \text{if } j\ge2.
  \end{cases}
\]
Thus, every interior cut point is assigned to the piece immediately to its left.
We say that an allocation uses at most $q$ cuts if there exist at most $q$ distinct interior cut points such that each bundle is a union of the induced pieces.
In configuration spaces below, formal cut coordinates may coincide.
Such coincident coordinates create only degenerate intervals and never increase the number of distinct cuts in the original cake.
Let $\kprop{n}$ denote the least number of cuts that always suffice for equal-length proportional cake cutting with $n$ agents.
Similarly, let $\kef{n}$ denote the least number of cuts that always suffice for equal-length envy-free cake cutting with $n$ agents.

The following exact-share theorem is the main tool in the peeling argument of \Cref{sec:ef-cut}.

\begin{theorem}[Stromquist--Woodall~\cite{stromquistwoodall1985}]\label{thm:stromquist-woodall}
Let $d\ge 2$, let $\eta_1,\dots,\eta_d$ be nonatomic probability measures on $[0,1]$, and let $\alpha\in[0,1]$.
Then there exists a measurable set $A\subseteq [0,1]$ such that $\eta_j(A)=\alpha$ for every $j\in\{1,\dots,d\}$.
Moreover, $A$ can be chosen as a union of pieces induced by at most $2d-2$ cuts.
\end{theorem}

\begin{remark}
Stromquist and Woodall originally state the cut-count guarantee for a circle.
In our interval notation, this means viewing $[0,1]$ as a circle by identifying $0$ and $1$.
The exact-share set can then be chosen as a union of at most $d-1$ connected pieces of the circle.
Opening the circle at the identified endpoint gives a subset of $[0,1]$ with at most $2d-2$ distinct interior boundary points.
Since all measures are nonatomic, assigning the opened endpoint to either side does not change any measure.
\end{remark}
\section{The Constrained Necklace-Splitting Theorem}\label{sec:necklace}

We use an interval form of constrained necklace splitting to balance the fractional assignments in the proportional proof.
Its preference form will also give the prime-power envy-free bound in \Cref{sec:ef-cut}.

\subsection{Configuration Space and Theorem Statement}\label{sec:theorem-statement}

Fix integers $m,r\ge 1$.
Let
$\Delta_m\coloneqq\left\{(y_1,\dots,y_m)\in\mathbb R_{\ge 0}^m : \sum_{k=1}^m y_k=1 \right\}$
be the simplex of possible lengths of $m$ consecutive subintervals.
For $y\in\Delta_m$, put $x_0(y)=0$ and $x_k(y)=\sum_{\ell=1}^k y_\ell$ for $k\in\{1,\dots,m\}$.
These cut positions form $m$ possibly degenerate consecutive intervals.
An \emph{$(m,r)$-configuration} is an equivalence class $\xi=[(y,f)]$, where $y\in\Delta_m$ and $f\colon\{1,\dots,m\}\to\{1,\dots,r\}$ labels the intervals; two pairs are identified when they differ only on indices $k$ with $y_k=0$.
Let $\mathcal C_{m,r}$ be this quotient space.
It is a compact metric space, being a finite simplicial quotient of $\Delta_m\times\{1,\dots,r\}^{\{1,\dots,m\}}$.

For a label $i$, let $V_i(\xi)$ be the union of the intervals labeled $i$, understood up to boundary points.
The bundles $V_1(\xi),\dots,V_r(\xi)$ partition the cake and together use at most $m-1$ cuts.
The symmetric group $S_r$ acts by relabeling
\[
  \sigma\cdot[(y,f)]=[(y,\sigma\circ f)]
  \quad\text{and}\quad
  V_{\sigma(i)}(\sigma\cdot\xi)=V_i(\xi)
\]
up to boundaries.

\begin{definition}
A \emph{preference matrix} on $\mathcal C_{m,r}$ is a family $\{A_i^j\}_{i,j\in\{1,\dots,r\}}$ of subsets of $\mathcal C_{m,r}$, where $A_i^j$ is the set of configurations at which agent $j$ accepts label $i$.
It is \emph{closed} if every $A_i^j$ is closed, \emph{proper} if $\bigcup_{i=1}^r A_i^j=\mathcal{C}_{m,r}$ for every agent $j$, and \emph{$S_r$-equivariant} if $\sigma\cdot A_i^j=A_{\sigma(i)}^j$ for every $\sigma\in S_r$ and every $i,j\in\{1,\dots,r\}$, where $\sigma\cdot A_i^j\coloneqq\{\sigma\cdot\xi:\xi\in A_i^j\}$.
\end{definition}

Thus $A_i^j$ records when agent $j$ accepts bundle $i$; closedness preserves acceptance under limits, properness guarantees an accepted label, and equivariance makes acceptance independent of label names.

The following result is Theorem~6.9 of Joji\'c, Panina, and \v{Z}ivaljevi\'c~\cite{jojic_panina_zivaljevic_2021}, translated to our relabeling convention.
On an interval, nonatomicity gives the required continuity.
For a prime power $r$, the theorem partitions the cake into $r$ labeled bundles, equalizes $t$ measures across them, and realizes an accepted assignment with $(r-1)(t+1)$ cuts.

\begin{theorem}[Theorem~6.9 of \cite{jojic_panina_zivaljevic_2021}]\label{thm:jpz69}
Let $r\ge 2$ be a prime power, let $t\ge 1$, and put $q\coloneqq (r-1)(t+1)$ and $m\coloneqq q+1$.
Let $\mu_1,\dots,\mu_t$ be nonatomic probability measures on $[0,1]$, and let $\{A^j_i\}_{i,j\in\{1,\dots,r\}}$ be a closed, $S_r$-equivariant, proper preference matrix on $\mathcal C_{m,r}$.
Then there exist a configuration $\xi\in \mathcal C_{m,r}$ and a permutation $\pi\in S_r$ such that $\xi\in A^j_{\pi(j)}$ for every $j\in\{1,\dots,r\}$, and $\mu_\tau(V_i(\xi))=1/r$ for every $\tau\in\{1,\dots,t\}$ and every $i\in\{1,\dots,r\}$.
In particular, the resulting partition uses at most $q=(r-1)(t+1)$ cuts.
\end{theorem}

\begin{corollary}\label{cor:jpz-diagonal}
Under the assumptions of \Cref{thm:jpz69}, there exists a configuration $\xi\in \mathcal C_{m,r}$ such that $\xi\in A^j_j$ for every $j\in\{1,\dots,r\}$, and $\mu_\tau(V_i(\xi))=1/r$ for every $\tau\in\{1,\dots,t\}$ and every $i\in\{1,\dots,r\}$.
The resulting partition uses at most $q=(r-1)(t+1)$ cuts.
\end{corollary}

\begin{proof}
Let $(\xi,\pi)$ be given by \Cref{thm:jpz69} and set $\xi'=\pi^{-1}\cdot\xi$.
Equivariance sends $A^j_{\pi(j)}$ to $A^j_j$.
Thus, $\xi'\in A^j_j$ for every $j$.
Relabeling does not change any measure equality or the cut count.
\end{proof}

This diagonal form assigns each agent an accepted label while balancing every auxiliary measure.
We also need the following continuity fact.

\begin{lemma}\label{lem:bundle-value-cont}
Let $\nu$ be a nonatomic finite measure on $[0,1]$, and fix a label $i\in\{1,\dots,r\}$.
Then the map $\xi\mapsto \nu(V_i(\xi))$ is continuous on $\mathcal C_{m,r}$.
\end{lemma}
\begin{proof}
Fix a label function $f\colon\{1,\dots,m\}\to\{1,\dots,r\}$.
On the fixed-label subset represented by $(y,f)$, let $x_0(y)=0$ and $x_k(y)=\sum_{\ell=1}^k y_\ell$.
Under the convention of \Cref{sec:preliminaries}, the first piece is $[x_0(y),x_1(y)]$ and the $k$th piece is $(x_{k-1}(y),x_k(y)]$ for $k\ge2$.
Hence $\nu(V_i(\xi))$ is a finite sum of differences $F_\nu(x_k(y))-F_\nu(x_{k-1}(y))$, where $F_\nu(x)=\nu([0,x])$.
Since $\nu$ is nonatomic, $F_\nu$ is continuous, and hence the map is continuous on each fixed-label subset.

If two representatives differ only by relabeling degenerate pieces, then the corresponding bundle unions differ only by sets of $\nu$-measure zero.
Thus the subsetwise formulas agree on overlaps and descend to the quotient $\mathcal C_{m,r}$.
Since the quotient is obtained by gluing finitely many fixed-label subsets, the pasting lemma gives continuity on $\mathcal C_{m,r}$.
\end{proof}

\subsection{Assignment Equalization}\label{sec:assignment-weight-equalization}

The transfer argument needs an equal-length configuration that equalizes the coordinates of a continuous equivariant map.
The following direct consequence of \Cref{cor:jpz-diagonal} supplies it; later, $F_i$ is the total fractional weight assigned to bundle $i$.

\begin{lemma}\label{lem:assignment-weight-equalization}
Let $r\ge 2$ be a prime power.
Suppose $F=(F_1,\dots,F_r)\colon \mathcal C_{2r-1,r}\to\mathbb R^r$ is continuous and satisfies $F_{\sigma(i)}(\sigma\cdot\xi)=F_i(\xi)$ for every $\sigma\in S_r$ and $i\in\{1,\dots,r\}$.
Then there exists $\xi\in\mathcal C_{2r-1,r}$ such that $\length{V_1(\xi)}=\cdots=\length{V_r(\xi)}=1/r$, and $F_1(\xi)=\cdots=F_r(\xi)$.
\end{lemma}

\begin{proof}
For each label $i\in\{1,\dots,r\}$, define
\[\textstyle
  A_i \coloneqq \left\{
    \xi\in\mathcal C_{2r-1,r} :
    F_i(\xi)\ge \max_{i'\in\{1,\dots,r\}}F_{i'}(\xi)
  \right\}.
\]
Let every agent have the same preference family, i.e., $A_i^j\coloneqq A_i$ for all $i,j\in\{1,\dots,r\}$.
Because $F$ is continuous, each set $A_i$ is closed.
Because at least one entry of $F(\xi)$ is maximal at every $\xi$, the matrix $\{A_i^j\}_{i,j\in\{1,\dots,r\}}$ is proper.
Because $F_{\sigma(i)}(\sigma\cdot\xi)=F_i(\xi)$, relabeling labels by any permutation $\sigma\in S_r$ sends maximizing labels at a configuration to maximizing labels at the relabeled configuration.
Thus the preference matrix is $S_r$-equivariant.

Apply \Cref{cor:jpz-diagonal} with the single measure $\lambda$ and with the preference matrix defined above.
We obtain a configuration $\xi\in\mathcal C_{2r-1,r}$ such that $\length{V_1(\xi)}=\cdots=\length{V_r(\xi)}=1/r$, and $\xi\in A_j^j=A_j$ for every $j\in\{1,\dots,r\}$.
Thus every label corresponds to a largest entry of $(F_1(\xi),\dots,F_r(\xi))$, and therefore $F_1(\xi)=\cdots=F_r(\xi)$.
\end{proof}
\section{Exact Cut Complexity for Proportionality}\label{sec:prop-cut}

We prove $\kprop{n}=2n-2$ for every $n$.
After an elementary lower bound, we write $n=rs$ with $r$ a prime power, partition the agents into $r$ balanced groups, and recurse within the corresponding top-level bundles.
The lower-bound instance is analogous to one used in the market-value setting \cite{barman_ebadian_latifian_shah_2025}.

\begin{proposition}\label{prop:elementary-lower-bound}
For every $n\ge 1$, we have $\kprop{n} \ge 2n-2$.
\end{proposition}

\begin{proof}
Fix $n\ge 1$.
The case $n=1$ is trivial.
We may therefore assume $n\ge 2$.
Consider the profile in which every agent $i$ has the same valuation $v_i(C')=n\cdot\length{C'\cap [0,1/n]}$ for $C'\subseteq \C$.
Let $(C_1,\dots,C_n)$ be any equal-length proportional allocation.
Since every agent values only the prefix $[0,1/n]$, each bundle $C_i$ must satisfy $\length{C_i\cap[0,1/n]}\ge 1/n^2$.
Summing over $i$ shows that equality must hold for every $i$, namely $\length{C_i\cap[0,1/n]}=1/n^2$ and $\length{C_i\cap(1/n,1]}=1/n-1/n^2$.
Therefore every bundle has positive length inside $[0,1/n]$ and positive length inside $(1/n,1]$.

Let $T$ be the set of distinct interior cut points used by the allocation.
The points of $T\cap(0,1/n)$ split $[0,1/n]$ into $|T\cap(0,1/n)|+1$ positive-length intervals, each contained in a single induced piece and therefore assigned to a single bundle.
Since every one of the $n$ bundles has positive length in $[0,1/n]$, these intervals must serve all $n$ bundles, and hence $|T\cap(0,1/n)|+1\ge n$.
Thus at least $n-1$ cuts lie in $(0,1/n)$.
This argument does not depend on whether $1/n$ itself is a cut point.
If an induced piece crosses $1/n$, its left part still belongs to only the one bundle containing that induced piece.

The same reasoning applied to $(1/n,1]$ gives $|T\cap(1/n,1)|\ge n-1$.
The two open intervals are disjoint.
Hence every such allocation uses at least $2n-2$ distinct interior cuts.
\end{proof}

\subsection{The Transfer Principle}\label{sec:transfer-principle}

The first step rounds a fractional assignment with balanced row and column sums, using Hall's marriage theorem, into an assignment of exactly $s$ agents to each of the $r$ labels.

\begin{lemma}\label{lem:balanced-support}
Let $M=(M_{ij})$ be an $r\times n$ nonnegative matrix such that every row sums to $1$ and every column sums to $1/s$, where $n=rs$ and $s$ is a positive integer.
Then there exists a map $\pi\colon\{1,\dots,n\}\to\{1,\dots,r\}$ with the following properties, where $\pi^{-1}(i)$ denotes the set of elements mapped to $i$.
\begin{enumerate}
  \item $|\pi^{-1}(i)| = s$ for every $i\in\{1,\dots,r\}$, and
  \item $M_{\pi(j),j}>0$ for every $j\in\{1,\dots,n\}$.
\end{enumerate}
\end{lemma}

\begin{proof}
Consider the bipartite graph with left vertices $\{1,\dots,r\}\times\{1,\dots,s\}$ and right vertices $\{1,\dots,n\}$, where $(i,\ell)$ is adjacent to $j$ if and only if $M_{ij}>0$.
We first show that this graph has a perfect matching.
The two sides both have $n=rs$ vertices.
Hence a matching covering the right side is perfect, and by Hall's theorem it is enough to show that every subset $S\subseteq\{1,\dots,n\}$ of right vertices has at least $|S|$ neighbors.

Let $T\subseteq\{1,\dots,r\}$ be the set of row indices that appear among the neighbors of $S$.
Then all positive entries in columns indexed by $S$ lie in rows indexed by $T$.
Summing the column sums over $S$ gives
\begin{align*}
  \frac{|S|}{s}
  &= \sum_{j\in S}\sum_{i=1}^r M_{ij}
   = \sum_{j\in S}\sum_{i\in T} M_{ij}
   \le \sum_{i\in T}\sum_{j=1}^n M_{ij}
   = |T|.
\end{align*}
Therefore, the neighborhood of $S$ in $\{1,\dots,r\}\times\{1,\dots,s\}$ has size $s|T|\ge |S|$.
Hall's theorem gives a perfect matching.

If column $j$ is matched to $(i,\ell)$, define $\pi(j)\coloneqq i$.
Because there are exactly $s$ copies of each row index $i$, we have $|\pi^{-1}(i)|=s$.
Because edges occur only where $M_{ij}>0$, we also have $M_{\pi(j),j}>0$ for every $j$.
\end{proof}

We use the following fractional assignments in the transfer.
The entry $M^\varepsilon_{ij}(\xi)$ is the weight that agent $j$ places on label $i$; positive weight is restricted to bundles worth more than $1/r-\varepsilon$.
The slack makes the support condition continuous and will vanish in a compactness argument.

\begin{definition}\label{def:admissible-support}
Let $n=rs$, where $s$ is a positive integer, and fix $\varepsilon>0$.
An \emph{$\varepsilon$-admissible fractional assignment scheme} for an $n$-agent profile is a continuous map $M^\varepsilon\colon\mathcal C_{2r-1,r}\to \mathbb R_{\ge 0}^{r\times n}$ with the following properties.
\begin{enumerate}
  \item For every configuration $\xi$ and every agent $j\in\{1,\dots,n\}$, the column sum satisfies $\sum_{i=1}^r M^\varepsilon_{ij}(\xi)=1/s$.
  \item The vector
        $F^\varepsilon(\xi)\coloneqq\left(\sum_{j=1}^nM^\varepsilon_{1j}(\xi),\dots,\sum_{j=1}^n M^\varepsilon_{rj}(\xi)\right)$
        satisfies
        $F^\varepsilon_{\sigma(i)}(\sigma\cdot\xi)    =F^\varepsilon_i(\xi)$
        for every $\sigma\in S_r$ and $i\in\{1,\dots,r\}$.
  \item Whenever $M^\varepsilon_{ij}(\xi)>0$, we have $v_j(V_i(\xi))>1/r-\varepsilon$.
\end{enumerate}
\end{definition}

The transfer lemma first equalizes the row sums and then applies Hall's theorem to the positive support.
A compactness argument removes the slack $\varepsilon$ while keeping the discrete assignment fixed along a subsequence.

\begin{lemma}\label{lem:support-transfer}
Let $n=rs$, where $r\ge 2$ is a prime power and $s$ is a positive integer, and let $v_1,\dots,v_n$ be an $n$-agent profile.
Suppose that for every $\varepsilon>0$ this profile admits an $\varepsilon$-admissible fractional assignment scheme.
Then this profile admits a partition of $\C$ into $r$ bundles $C_1,\dots,C_r$ with at most $2r-2$ cuts together with a map $\pi\colon\{1,\dots,n\}\to\{1,\dots,r\}$ such that
\begin{enumerate}
  \item $\length{C_1}=\cdots=\length{C_r}=1/r$,
  \item $|\pi^{-1}(i)| = s$ for all $i\in\{1,\dots,r\}$, and
  \item $v_j(C_{\pi(j)})\ge 1/r$ for every agent $j\in\{1,\dots,n\}$.
\end{enumerate}
\end{lemma}

\begin{proof}
Fix $\varepsilon>0$, and choose an $\varepsilon$-admissible fractional assignment scheme $M^\varepsilon$.
Let $F^\varepsilon$ be its vector of total assignment weights.
Since $M^\varepsilon$ is continuous, $F^\varepsilon$ is continuous as well.
By \Cref{lem:assignment-weight-equalization}, there exists $\xi_\varepsilon\in\mathcal C_{2r-1,r}$ such that $\length{V_1(\xi_\varepsilon)}=\cdots=\length{V_r(\xi_\varepsilon)}=1/r$ and $F^\varepsilon_1(\xi_\varepsilon)=\cdots=F^\varepsilon_r(\xi_\varepsilon)$.

The sum of all entries of $M^\varepsilon(\xi_\varepsilon)$ is $n/s=r$.
Since all $r$ row sums are equal, each equals $1$.

Applying \Cref{lem:balanced-support} to $M^\varepsilon(\xi_\varepsilon)$, we obtain a map $\pi_\varepsilon\colon\{1,\dots,n\}\to\{1,\dots,r\}$ such that each label has exactly $s$ preimages and $M^\varepsilon_{\pi_\varepsilon(j),j}(\xi_\varepsilon)>0$ $(j\in\{1,\dots,n\})$.
By admissibility, $v_j(V_{\pi_\varepsilon(j)}(\xi_\varepsilon))>1/r-\varepsilon$.

Let $\Pi$ be the finite set of maps $\pi\colon\{1,\dots,n\}\to\{1,\dots,r\}$ that assign $s$ agents to each label.
Choose a sequence $\varepsilon_k>0$ with $\varepsilon_k\to0$, apply the preceding argument to each $\varepsilon_k$, and denote the resulting pair by $(\xi_k,\pi_k)$.
Since $\Pi$ is finite, pass to a subsequence with $\pi_k=\pi$ for all $k$.
Since $\mathcal C_{2r-1,r}$ is a compact metric space, pass to a further subsequence with $\xi_k\to\xi$.

Set $C_i\coloneqq V_i(\xi)$.
These sets partition $\C$ up to boundary points, which we assign according to the half-open convention of \Cref{sec:preliminaries}.
Nonatomicity ensures that this choice does not affect any length or value.
By \Cref{lem:bundle-value-cont}, $\xi\mapsto v_j(V_{\pi(j)}(\xi))$ is continuous for each $j$.
Since $v_j(V_{\pi(j)}(\xi_k))>1/r-\varepsilon_k$ and $\varepsilon_k\to0$, we obtain $v_j(C_{\pi(j)})\ge1/r$.
The same continuity lemma applied to $\lambda$ gives $\length{C_i}=1/r$ $(i\in\{1,\dots,r\})$.
Since $\xi\in\mathcal C_{2r-1,r}$, it has at most $2r-2$ formal interior cut positions.
Repeated positions create only degenerate intervals.
Thus, the partition uses at most $2r-2$ distinct cuts.
\end{proof}

For proportionality, an admissible scheme is obtained by normalizing each agent's positive excesses above $1/r-\varepsilon$.
This assigns positive weight only to sufficiently valuable bundles.

\begin{lemma}\label{lem:prop-support-scheme}
Let $n=rs$, where $r\ge 2$ is a prime power and $s$ is a positive integer, and let $v_1,\dots,v_n$ be an $n$-agent profile.
Then for every $\varepsilon>0$, this profile admits an $\varepsilon$-admissible fractional assignment scheme.
\end{lemma}

\begin{proof}
Fix $\varepsilon>0$.
For each configuration $\xi\in\mathcal C_{2r-1,r}$, bundle index $i\in\{1,\dots,r\}$, and agent $j\in\{1,\dots,n\}$, define $g_{ij}(\xi)\coloneqq \max\{0,\ v_j(V_i(\xi))-(1/r-\varepsilon)\}$.
The maps $\xi\mapsto v_j(V_i(\xi))$ are continuous on $\mathcal C_{2r-1,r}$ by \Cref{lem:bundle-value-cont}.
Hence each $g_{ij}$ is continuous.
For every fixed $j$, the bundle values sum to $1$.
Hence at least one bundle has value at least $1/r$, and therefore $d_j(\xi)\coloneqq\sum_{k=1}^r g_{kj}(\xi)$ is continuous and strictly positive for every $\xi$.

Now define $M^\varepsilon_{ij}(\xi)\coloneqq (r/n)\,g_{ij}(\xi)/d_j(\xi)$.

This is a well-defined continuous nonnegative map.
For each $j\in\{1,\dots,n\}$ we have $\sum_{i=1}^r M^\varepsilon_{ij}(\xi)=r/n=1/s$.
If $M^\varepsilon_{ij}(\xi)>0$, then $g_{ij}(\xi)>0$, hence $v_j(V_i(\xi))>1/r-\varepsilon$.
Finally, relabeling sends row $i$ of $M^\varepsilon(\xi)$ to row $\sigma(i)$ of $M^\varepsilon(\sigma\cdot\xi)$ because $g_{\sigma(i),j}(\sigma\cdot\xi)=g_{ij}(\xi)$, and therefore $d_j(\sigma\cdot\xi)=d_j(\xi)$ and $M^\varepsilon_{\sigma(i),j}(\sigma\cdot\xi)=M^\varepsilon_{ij}(\xi)$.
Summing over $j$ gives $F^\varepsilon_{\sigma(i)}(\sigma\cdot\xi)=F^\varepsilon_i(\xi)$.
\end{proof}

Combining \Cref{lem:prop-support-scheme,lem:support-transfer} gives $r$ equal-length top-level bundles with at most $2r-2$ cuts, exactly $s$ assigned agents per bundle, and value at least $1/r$ for every assigned agent.

\subsection{Recurrence and Induction}\label{sec:recurrence-induction}

Each top-level bundle may be disconnected.
Concatenating its components lets us solve a normalized problem on $[0,1]$ and pull the solution back without adding cuts at the component boundaries.

Let $C\subseteq[0,1]$ be a positive-length finite union of intervals whose interior component boundaries have already been counted as cuts.
Concatenate its components from left to right and rescale by $\length{C}$ to obtain $\phi\colon C\to[0,1]$.
For every finite measure $\nu$ with $\nu(C)>0$, put $\widehat\nu(Y)=\nu(\phi^{-1}(Y))/\nu(C)$; a measure vanishing on $C$ imposes no constraint on the subproblem.

\begin{lemma}\label{lem:concatenation-pullback}
If the normalized subproblem has a solution using at most $k$ cuts, then the pulled-back solution uses at most $k$ new cuts in the original cake inside $C$.
Lengths are multiplied by $\length{C}$, and values for each measure $\nu$ with $\nu(C)>0$ are multiplied by $\nu(C)$.
Thus equal-length conditions, value inequalities, envy comparisons, and equality constraints in the normalized solution become the corresponding conditions inside $C$.
\end{lemma}

\begin{proof}
Each normalized cut pulls back either to one new point inside a component of $C$ or to an already counted component boundary.
Hence, at most $k$ new cuts are created.
Moreover, we have
\[
  \lambda(\phi^{-1}(Y))=\length{C}\cdot\lambda(Y)
  \quad\text{and}\quad
  \nu(\phi^{-1}(Y))=\nu(C)\cdot\widehat\nu(Y).
\]
The two positive scale factors give all claimed transfers.
Endpoint choices are immaterial by nonatomicity.
\end{proof}

Applying the top-level partition and assignment together with this lemma gives the recurrence.

\begin{lemma}\label{lem:prop-rec}
Let $n=rs$, where $r\ge 2$ is a prime power and $s$ is a positive integer.
Then $\kprop{n}\le (2r-2)+r\kprop{s}$.
\end{lemma}

\begin{proof}
Fix an arbitrary $n$-agent profile.
By \Cref{lem:prop-support-scheme}, it has an admissible scheme for every $\varepsilon>0$.
Applying \Cref{lem:support-transfer}, we obtain bundles $C_1,\dots,C_r$ of length $1/r$ and agent groups $S_i\coloneqq\pi^{-1}(i)$ $(i\in\{1,\dots,r\})$, each of size $s$, such that every $j\in S_i$ satisfies $v_j(C_i)\ge 1/r$.

Fix $i\in\{1,\dots,r\}$.
Restrict and normalize each $v_j$ for $j\in S_i$ on $C_i$; this is valid because $v_j(C_i)\ge1/r>0$.
By the definition of $\kprop{s}$ and \Cref{lem:concatenation-pullback}, the normalized instance has a pulled-back solution using at most $\kprop{s}$ new cuts, with bundle length $1/n$.
Each assigned agent $j\in S_i$ receives value at least $v_j(C_i)/s\ge (1/r)/s=1/n$.

Applying this independently inside each bundle gives an equal-length proportional allocation of the whole cake.
The top-level partition contributes at most $2r-2$ cuts, and the recursive cuts add across the $r$ disjoint bundles.
Hence the total number of cuts is at most $(2r-2)+r\kprop{s}$.
\end{proof}

\begin{theorem}\label{thm:main}
For every $n\ge 1$, we have $\kprop{n} = 2n-2$.
\end{theorem}

\begin{proof}
We prove the upper bound $\kprop{n}\le 2n-2$ by strong induction on $n$.

The case $n=1$ is trivial.
Now let $n\ge 2$, and assume for every $t<n$ that $\kprop{t}\le 2t-2$.

Choose any prime-power divisor $r>1$ of $n$, and write $n=rs$.
By \Cref{lem:prop-rec} and the induction hypothesis applied to $s<n$, we obtain $\kprop{n}\le (2r-2)+r\kprop{s}\le (2r-2)+r(2s-2)=2rs-2=2n-2$.
This proves $\kprop{n}\le 2n-2$ for every $n$.
Together with \Cref{prop:elementary-lower-bound}, this gives $\kprop{n}=2n-2$.
\end{proof}
\section{Cut Complexity for Envy-Freeness}\label{sec:ef-cut}

Since envy-freeness implies proportionality, \Cref{prop:elementary-lower-bound} gives $\kef{n}\ge2n-2$.
We match this bound when $n$ is a prime power and prove $\kef{n}=O(n^{1.525})$ in general.
For the general upper bound, we repeatedly apply the Stromquist--Woodall exact-share theorem (\Cref{thm:stromquist-woodall}), treating valuations assigned to the other side of a recursive split as auxiliary measures.
These measures control envy across the split.

\subsection{Prime Powers}\label{sec:prime-power-consequence}

For fixed $m$ and valuations $v_1,\dots,v_r$, define a preference matrix on $\mathcal C_{m,r}$ by
\[\textstyle
  A_i^j
  \coloneqq
  \left\{
    \xi\in\mathcal C_{m,r} :
    v_j(V_i(\xi))\ge \max_{i'\in\{1,\dots,r\}}v_j(V_{i'}(\xi))
  \right\}.
\]
The matrix is closed by \Cref{lem:bundle-value-cont}, equivariant under relabeling, and proper because every agent has a maximizing bundle.
Hence \Cref{cor:jpz-diagonal} applies.

\begin{proposition}\label{prop:prime-power-direct}
If $r\ge2$ is a prime power, then $\kef{r}=2r-2$.
\end{proposition}

\begin{proof}
Apply \Cref{cor:jpz-diagonal} with $t=1$, measure $\mu_1=\lambda$, and the preference matrix above.
The resulting bundles have length $1/r$, and agent $j$ weakly prefers its assigned bundle $V_j$ to every $V_i$.
Thus the allocation is envy-free and uses at most $2r-2$ cuts.
Conversely, envy-freeness implies proportionality.
Thus, \Cref{thm:main} gives $\kef{r}\ge\kprop{r}=2r-2$.
\end{proof}

\subsection{The Exact-Share Recurrence}

For $n$ agents and $h$ additional nonatomic probability measures $\nu_1,\dots,\nu_h$, let $\kefh{n}{h}$ be the least number of cuts that always suffices for an equal-length envy-free allocation in which each $\nu_\ell$ assigns the same value to all $n$ bundles.
The additional measures are auxiliary: they impose equality constraints but receive no bundles.
Thus $\kef{n}=\kefh{n}{0}$, and $\kefh{1}{h}=0$ because the unique bundle is the whole cake.

Exact-share splitting yields the following recurrence.
We apply \Cref{thm:stromquist-woodall} to length, the auxiliary measures, and all agent valuations, obtaining a region to which every measure assigns the prescribed share.

\begin{theorem}\label{thm:ef-additive}
For every $n\ge 2$, every $h\ge 0$, and every $a\in\{1,\dots,n-1\}$,
\[
  \kefh{n}{h}\le 2(n+h)+\kefh{a}{h+n-a}+\kefh{n-a}{h+a}.
\]
\end{theorem}

\begin{proof}
Let the agent valuations be $\mu_1,\dots,\mu_n$, let the auxiliary measures be $\nu_1,\dots,\nu_h$, and let $\lambda$ denote length.
Apply the Stromquist--Woodall theorem, \Cref{thm:stromquist-woodall}, to the $m=n+h+1$ probability measures $\lambda,\nu_1,\dots,\nu_h,\mu_1,\dots,\mu_n$ with target ratio $\alpha=a/n$.
We obtain a set $A\subseteq [0,1]$ using at most $2m-2=2(n+h)$ cuts such that $\length{A}=a/n$, $\nu_j(A)=a/n$ for every $j$, and $\mu_i(A)=a/n$ for every $i$.
Let $B\coloneqq[0,1]\setminus A$.
Then $\length{B}=(n-a)/n$, $\nu_j(B)=(n-a)/n$ for every $j$, and $\mu_i(B)=(n-a)/n$ for every $i$.

The component boundaries of $A$ and $B$ are already cut by the initial exact-share split.
Thus \Cref{lem:concatenation-pullback} applies to both sides, and only the recursive cuts themselves add new cuts in the original cake.

Choose any $S\subseteq\{1,\dots,n\}$ with $|S|=a$.
On $A$, the agents are $\{\mu_i:i\in S\}$, and the auxiliary measures are $\nu_1,\dots,\nu_h$ together with $\{\mu_i:i\notin S\}$.
After restriction, normalization, and pullback, this subproblem costs at most $\kefh{a}{h+n-a}$ new cuts.
Every resulting bundle has length $1/n$, and each auxiliary measure---including every $\mu_i$ with $i\notin S$---assigns it value $1/n$.

On $B$, the agents are $\{\mu_i:i\notin S\}$, and the auxiliary measures are $\nu_1,\dots,\nu_h$ together with $\{\mu_i:i\in S\}$.
This subproblem costs at most $\kefh{n-a}{h+a}$ new cuts and has the analogous length and value properties.

To verify global envy-freeness, fix $i\in S$ and let $X_i\subseteq A$ be the bundle assigned to $i$.
Local envy-freeness implies $\mu_i(X_i)\ge\mu_i(A)/a=1/n$.
Since $\mu_i$ is auxiliary on $B$, every bundle $Y\subseteq B$ has value $\mu_i(Y)=1/n$.
Thus agent $i$ envies neither a bundle in $A$ nor one in $B$.
The case $i\notin S$ is symmetric, proving the recurrence.
\end{proof}

\subsection{Peeling and a Subquadratic Upper Bound}

We iterate the exact-share recurrence with $a=1$, peeling off one bundle at a time until only a prime-power core remains.

\begin{theorem}\label{thm:ef-peeling}
Let $n\ge 2$, let $2\le r\le n$ be a prime power, and define $t\coloneqq n-r$.
Then, we have $\kef{n}\le 2nt+(r-1)(t+2)\le 3nt+2n$.
\end{theorem}

\begin{proof}
Using $\kefh{1}{m}=0$ for all $m\ge 0$, taking $a=1$ in \Cref{thm:ef-additive} gives
\[
  \kefh{n}{h}\le 2(n+h)+\kefh{n-1}{h+1}.
\]
Starting from $\kefh{r+t}{0}$, we apply this recurrence repeatedly while decreasing the number of agents by one and increasing the number of auxiliary measures by one.
Since $r+t=n$, for each $k=0,\dots,t-1$ we obtain
\[
  \kefh{n-k}{k}=\kefh{r+t-k}{k}\le 2(r+t)+\kefh{r+t-k-1}{k+1}=2n+\kefh{n-k-1}{k+1},
\]
since the sum of the two indices stays equal to $n$ throughout.
Summing these $t$ inequalities gives
\[
  \kef{n}\le 2nt+\kefh{r}{t}.
\]
To bound the remaining term, fix any prime-power core instance with $r$ agent valuations and $t$ auxiliary measures $\nu_1,\dots,\nu_t$.
Apply \Cref{cor:jpz-diagonal} with the $t+1$ measures $\lambda,\nu_1,\dots,\nu_t$, and the preference family induced by the $r$ agent valuations.
The length measure makes every core bundle have length $1/r$, each auxiliary measure assigns the same value to all core bundles, and the maximizing preferences make the diagonal assignment envy-free.
Since the number of measures is $t+1$, the cut bound from \Cref{cor:jpz-diagonal} is $(r-1)(t+2)$.
Hence
\[
  \kefh{r}{t}\le (r-1)(t+2).
\]
This proves the first displayed bound.
Since $r=n-t$, we also obtain
\[
  \kef{n}\le 2nt+(n-t-1)(t+2)=3nt-t^2-3t+2n-2,
\]
and dropping the negative terms gives $\kef{n}\le 3nt+2n$.
\end{proof}

To obtain a subquadratic bound, choose a prime power $r\le n$ with a small gap $t=n-r$.
Baker, Harman, and Pintz~\cite{baker_harman_pintz_2001} show that $[x-x^{0.525},x]$ contains a prime for all sufficiently large $x$.

\begin{proposition}\label{prop:ef-bhp}
For all sufficiently large $n$, we have $\kef{n}\le 3n^{1.525}+2n$.
In particular, $\kef{n}=O(n^{1.525})=o(n^2)$.
\end{proposition}

\begin{proof}
By Baker, Harman, and Pintz~\cite{baker_harman_pintz_2001}, for all sufficiently large $n$ the interval $[n-n^{0.525},n]$ contains a prime $r$.
Thus $r$ is a prime power and $t=n-r\le n^{0.525}$.
Substituting into \Cref{thm:ef-peeling} yields
\[
  \kef{n}\le 3n\cdot n^{0.525}+2n = 3n^{1.525}+2n=O(n^{1.525}).
\]
\end{proof}

\begin{remark}
The proof uses a number-theoretic result only to bound the gap $t=n-r$.
Accordingly, any stronger bound on the distance from $n$ to a prime power improves the all-$n$ bound from \Cref{thm:ef-peeling}.
Concretely, if every sufficiently large $n$ admits a prime power $r\in[n-\Delta(n),n]$, then $t=n-r\le\Delta(n)$ and \Cref{thm:ef-peeling} yields $\kef{n}\le 3n\Delta(n)+2n$.
Under the Riemann hypothesis, Cram\'er's classical prime-gap estimate gives $\Delta(n)=O(\sqrt n\log n)$ and therefore $\kef{n}=O(n^{3/2}\log n)$ \cite{cramer_1936,dudek_2015}.
The scale $(\log n)^2$ is a classical conjectural benchmark for prime gaps, originating in Cram\'er's probabilistic model \cite{cramer_1936}.
Under a Cram\'er-type hypothesis giving $\Delta(n)=O((\log n)^2)$, we would obtain $\kef{n}=O\bigl(n(\log n)^2\bigr)$.
\end{remark}

\section{Query Complexity in the Robertson--Webb Model}\label{sec:uncomputable}

In this section, we show that no deterministic protocol in the Robertson--Webb model terminates after finitely many queries on every instance and always returns an exact equal-length proportional allocation.
We then show that proportionality can be approximated to arbitrary accuracy with finitely many queries while maintaining exact lengths.

Since length is public, the protocol may place cuts at specified coordinates and perform exact arithmetic on coordinates without querying an agent.
Its only access to a private valuation $v_i$ is through $\textsc{Eval}_i(x)$, which returns $v_i([0,x])$, and $\textsc{Cut}_i(r)$, which returns the leftmost $x$ satisfying $v_i([0,x])=r$.
These queries also give access to arbitrary intervals.
The value of $[a,b]$ is obtained by subtracting $\textsc{Eval}_i(a)$ from $\textsc{Eval}_i(b)$.
To find the leftmost $b\ge a$ satisfying $v_i([a,b])=r$, the protocol first obtains $v_i([0,a])$ and then makes the query $\textsc{Cut}_i(v_i([0,a])+r)$.
Only queries for which the requested cut exists are made.
Protocols output finite unions of intervals whose endpoints are either returned by $\textsc{Cut}$ queries or specified by their coordinates.
\subsection{Finite-Query Impossibility}\label{sec:rw-impossibility}

For two agents with the same valuation $v$, an equal-length allocation has the form $(A,[0,1]\setminus A)$ with $\length{A}=1/2$.
It is proportional if and only if
\[
  v(A)=v([0,1]\setminus A)=\frac12.
\]
Thus, the problem is to halve simultaneously the known length measure $\lambda$ and an unknown measure $v$.
Exact division of two unknown measures cannot in general be accomplished by a finite algorithm \cite[Thm.~2.2]{robertson_webb_1997}.
We adapt their decision-tree argument to show that the impossibility remains when length is known.

\begin{theorem}\label{thm:main-un}
There is no deterministic protocol in the Robertson--Webb model that terminates after finitely many queries on every instance and always returns an exact equal-length proportional allocation.
This remains true for two agents with the same nonatomic valuation.
\end{theorem}

\begin{proof}
Suppose that such a protocol exists, and fix it for two agents.
We answer its queries adversarially while constructing a nonatomic valuation $v$, which will be shared by the two agents.
Begin with $\mathcal P^0=\{[0,1]\}$ and assign value $1$ to its sole member.
After the $j$th query, let $\mathcal P^j$ be the partition obtained from $\mathcal P^{j-1}$ by adding every new endpoint introduced by that query.
At every stage, each member of $\mathcal P^j$ has positive value, and no $\mathcal U\subseteq\mathcal P^j$ satisfies $\lambda(\bigcup_{U\in\mathcal U}U)=v(\bigcup_{U\in\mathcal U}U)=1/2$.
Both properties hold initially.

It suffices to show that both properties are preserved whenever a new endpoint is added to the current partition $\mathcal P$.
Suppose that the new endpoint splits a subinterval $I\in\mathcal P$ into two subintervals $I'$ and $I''$ of positive length, producing $\mathcal P'$.
If the protocol specifies the splitting point, the lengths of $I'$ and $I''$ are determined, but the adversary can choose how the value of $I$ is divided between them.
If the splitting point is the answer to a $\textsc{Cut}$ query, the query determines how the value is divided, but the adversary can choose the lengths by selecting the point.
Write $d=v(I)/\lambda(I)$ for the average density on $I$.
In the first case, the adversary chooses $v(I')\in(0,v(I))$ and sets $v(I'')=v(I)-v(I')$.
In the second case, it chooses $\lambda(I')\in(0,\lambda(I))$ and sets $\lambda(I'')=\lambda(I)-\lambda(I')$.
Fix $\mathcal U\subseteq\mathcal P'$.
If $\mathcal U$ contains both $I'$ and $I''$, replacing them by $I$ gives a subcollection of $\mathcal P$ whose union has the same length and value.
If $\mathcal U$ contains neither, then $\mathcal U\subseteq\mathcal P$.
By the second property for $\mathcal P$, neither case yields both length and value $1/2$.
It remains to consider a subcollection $\mathcal U$ containing exactly one of $I'$ and $I''$.
If the protocol specified the splitting point, then $\lambda(\bigcup_{U\in\mathcal U}U)$ is fixed.
If this length is not $1/2$, the union cannot have both length and value $1/2$ for any choice of $v(I')$.
If it is $1/2$, the equation $v(\bigcup_{U\in\mathcal U}U)=1/2$ excludes at most one choice of $v(I')$.
If the splitting point answers a $\textsc{Cut}$ query, the same argument applies with length and value interchanged.
Thus, only finitely many values of the quantity chosen by the adversary are excluded.
We choose this quantity close to the value that would give both new subintervals the same density $d$ as $I$.
At the $h$th time a new endpoint is added to the partition, we choose $v(I')$ close to $d\lambda(I')$ in the first case and $\lambda(I')$ close to $v(I')/d$ in the second.
We choose it close enough that the density of each new subinterval lies between $2^{-2^{-h}}d$ and $2^{2^{-h}}d$.
The target value lies strictly inside the relevant interval, and only finitely many choices are excluded.
Thus, such a choice exists.
It preserves both properties.

This rule determines the adversary's answers.
For an $\textsc{Eval}_i(x)$ query, split at $x$ if necessary and return the total value of the subintervals to its left.
For a $\textsc{Cut}_i(q)$ query, return an existing endpoint if the cumulative value there is $q$.
Otherwise, $q$ lies strictly between the cumulative values at the endpoints of a unique subinterval because every member of the partition has positive value.
The query fixes the value of the left part of this subinterval, and the adversary chooses its length, and hence the returned point, by the rule above.

If the protocol announces an output, add its endpoints to the current partition one at a time using the same refinement rule.
The values assigned during this construction are consistent because a refinement divides the value of one subinterval without changing the value of any earlier subinterval.
If only finitely many new endpoints are added, assign a constant density to each member of the eventual partition.
If infinitely many are added, let $f_h$ be the density that is constant on each member of the partition after the $h$th new endpoint has been added and realizes its assigned value.
When the $h$th new endpoint is added, the density of each new subinterval differs from the density on the subinterval that was split by a factor between $2^{-2^{-h}}$ and $2^{2^{-h}}$.
Since $\sum_{h\ge1}2^{-h}=1$, every $f_h$ takes values between $1/2$ and $2$.
For $m>h$, the ratio $f_m/f_h$ lies uniformly between $2^{-\sum_{j>h}2^{-j}}$ and $2^{\sum_{j>h}2^{-j}}$.
Hence, $(f_h)$ is uniformly Cauchy and converges to a density $f$ with the same bounds.
Thus, in either case the construction defines a single nonatomic valuation $v$ consistent with every answer given to the protocol.
Its positive density also makes every answer to a $\textsc{Cut}$ query the unique leftmost solution.

Because the protocol is deterministic, its run on $v$ reproduces the adversarial answer sequence.
If this sequence is infinite, the protocol does not terminate on $v$, contradicting the assumption.
Otherwise, the protocol produces an output after finitely many queries, and its bundles are unions of members of the partition obtained after adding their endpoints.
The second property shows that neither bundle can have both length and value $1/2$.
Hence, the output is not an equal-length proportional allocation, a contradiction.
\end{proof}

\subsection{Approximation with Finitely Many Queries}\label{sec:rw-approximation}

Although exact computation is impossible, proportionality can be approximated with finitely many queries.
For $\varepsilon>0$, call an equal-length allocation \emph{$\varepsilon$-proportional} if $v_i(C_i)\ge 1/n-\varepsilon$ for every agent $i$.

\begin{theorem}\label{thm:rw-approx}
For every $\varepsilon>0$, there is a protocol in the Robertson--Webb model that computes an equal-length $\varepsilon$-proportional allocation using $O(n^2/\varepsilon)$ queries and at most $2n-2$ cuts.
\end{theorem}

\begin{proof}
The case $n=1$ is trivial.
We may therefore assume $n\ge2$.
If $\varepsilon\ge1/n$, any partition into $n$ intervals of length $1/n$ suffices without queries.
We may therefore assume $0<\varepsilon<1/n$ and put $q=2n-2$.
Set $B=\lceil 2q/\varepsilon\rceil$.
For each agent $i$ and each $r=1,\dots,B-1$, ask $\textsc{Cut}_i(r/B)$.
Together with $0$ and $1$, the returned points partition the cake into $B$ intervals, each of value $1/B$ to agent $i$.
We call this partition agent $i$'s grid.
Let $G=\{I_1,\dots,I_s\}$ be the interval partition whose endpoints are all the points in the $n$ agents' grids, with the intervals listed from left to right.
Then $s\le n(B-1)+1\le nB$.
Every $I_r$ is contained in one interval of each agent's grid and therefore has value at most $1/B$ to every agent.

Enumerate all combinatorial types of allocations with $q$ cuts.
Such a type consists of indices $1\le r_1\le\dots\le r_q\le s$ and an agent label for each of the $q+1$ intervals induced by the cuts.
It represents allocations whose $j$th cut lies in the closure of $I_{r_j}$.
At a grid endpoint, the same placement may therefore be represented by either adjacent interval.
Equal consecutive indices are allowed because two cuts may lie in the same interval of $G$, and the interval between them may have positive length.
Allowing coincident cuts also includes allocations with fewer than $q$ cuts.
Once a type is fixed, the cut positions are the only variables.
The conditions that the $j$th cut lies in the closure of $I_{r_j}$, that the cuts remain ordered, and that every labeled bundle has length exactly $1/n$ form a finite system of linear equalities and inequalities.

For any feasible placement $x_1\le\dots\le x_q$ of a fixed type, let $C_i$ be the union of the intervals labeled $i$.
For each cut $x_j$ in the closure of $I_{r_j}$, agent $i$'s grid determines an integer $a_{ij}$ such that $a_{ij}/B\le v_i([0,x_j])\le (a_{ij}+1)/B$.
Let $P_h$ be the piece between $x_h$ and $x_{h+1}$ under the convention of \Cref{sec:preliminaries}, where $x_0=0$ and $x_{q+1}=1$.
The value of $C_i$ is the sum of $v_i(P_h)=v_i([0,x_{h+1}])-v_i([0,x_h])$ over the pieces $P_h$ labeled $i$.
Replace each $v_i([0,x_j])$ in this sum by $a_{ij}/B$, and denote the result by $\widehat v_i$.
This number is determined by the type and the queried grids.
Each cut contributes to the error at most once.
Hence,
\begin{align}
  \bigl|v_i(C_i)-\widehat v_i\bigr| &\le \frac qB. \label{eq:grid-error}
\end{align}

By \Cref{thm:main}, there is an exact equal-length proportional allocation with at most $q$ cuts.
Its combinatorial type occurs in the enumeration and, by \eqref{eq:grid-error}, satisfies $\widehat v_i\ge1/n-q/B$ for every agent $i$.
Select any such type and output any feasible placement of its cuts.
Applying \eqref{eq:grid-error} once more gives
\begin{align*}
  v_i(C_i) &\ge \frac1n-\frac{2q}{B}\ge \frac1n-\varepsilon
  \qquad(i=1,\dots,n).
\end{align*}
The protocol uses $n(B-1)=O(n^2/\varepsilon)$ queries and at most $q=2n-2$ cuts.
\end{proof}

The theorem concerns query complexity only, since enumerating the combinatorial types may take exponential time.
Allowing more cuts gives the following polynomial-time alternative.

\begin{proposition}\label{prop:rw-approx-polytime}
For every $\varepsilon>0$, an equal-length $\varepsilon$-proportional allocation can be computed in polynomial time in the real-RAM oracle model using $O(n^3/\varepsilon)$ Robertson--Webb queries and $O(n^2/\varepsilon)$ cuts.
\end{proposition}

\begin{proof}
The case $n=1$ is immediate.
We may therefore assume $n\ge2$.
If $\varepsilon\ge1/n$, any partition into $n$ intervals of length $1/n$ suffices without queries.
We may therefore assume $0<\varepsilon<1/n$.
Construct $B$ and $G=\{I_1,\dots,I_s\}$ as in the proof of \Cref{thm:rw-approx}, and query every agent's value for every $I_r$.
This uses $O(ns)=O(n^3/\varepsilon)$ queries.
Introduce a nonnegative variable $y_{ir}$ for the fraction of the length of $I_r$ assigned to agent $i$.
Impose $\sum_i y_{ir}=1$ for every $r$, $\sum_r \length{I_r}y_{ir}=1/n$ for every $i$, and $\sum_r v_i(I_r)y_{ir}\ge 1/n$ for every $i$.
The assignment $y_{ir}=1/n$ is feasible, and a basic feasible solution has at most $s+2n$ positive variables.
Since every $I_r$ has at least one positive variable, there are at most $2n$ positive variables beyond the first one for each interval.
Subdividing an interval requires one fewer cut than the number of its positive variables.
Thus, the total number of additional cuts is at most $2n$.
Dividing each such interval into subintervals of the prescribed lengths decreases each agent's value from its fractional value by at most $2n/B\le\varepsilon$, because $v_i(I_r)\le1/B$.
The allocation uses at most $s-1+2n=O(n^2/\varepsilon)$ cuts.
\end{proof}

\section{Two-Knife Augmentation}\label{sec:moving-knife}

We augment the Robertson--Webb model with two-knife queries.

We use the following two additional types of queries and count them separately from Robertson--Webb queries.
We assume that the two knives can be moved continuously and monotonically while maintaining the specified length or value.
A \textsc{MoveLength} query specifies a fixed length, initial and final knife positions, a value threshold for each of a set of agents, and an integer $k$.
The two knives move continuously and monotonically while the portion of the current subcake between them keeps the specified length.
Each agent raises a hand whenever that portion is worth at least the agent's threshold.
If fewer than $k$ hands are raised initially and at least $k$ are raised finally, the query stops at the first time when at least $k$ hands are raised and returns the two knife positions.
A \textsc{MoveValue}$_i$ query specifies a value $\alpha$, a target length $\ell$, initial positions of both knives at which the intervening portion is worth $\alpha$, and a final position for one of the knives.
Agent $i$ moves the specified knife continuously and monotonically toward its final position and moves the other knife continuously and monotonically with it so that the intervening portion remains worth $\alpha$.
If $\ell$ lies between the initial length and the length when the specified knife reaches its final position, the query stops when the intervening portion first has length $\ell$ and returns the two knife positions.
If a motion or answer is not unique, any one may be used, and the procedure must be correct for every choice.
Knife positions and query endpoints count as cuts only when they form boundaries of the final allocation.

Queries may also be applied to a disconnected subcake $C=J_1\cup\cdots\cup J_q$ of length $L$, where the intervals are pairwise disjoint and listed from left to right.
Removing the gaps and concatenating the components gives a copy of $C$ on $[0,L]$.
For $0\le x\le y\le L$, let $C\langle x,y\rangle$ be the portion corresponding to $(x,y]$ in this copy, with the left endpoint included when $x=0$, as in the convention of \Cref{sec:preliminaries}.
This portion has length $y-x$ and at most $q$ components in the original cake.
Its value can be computed using $O(q)$ \textnormal{\textsc{Eval}} queries, and given one endpoint and a target value, the other endpoint can be found with $O(q)$ \textnormal{\textsc{Eval}} queries and at most one \textnormal{\textsc{Cut}} query.
Crossing a gap leaves the corresponding portion unchanged.
Thus, a continuous monotone motion on $[0,L]$ also gives one on the original cake.

\subsection{Balanced Halving}

We show that one \textsc{MoveLength} query splits an even subinstance into two equal subinstances while preserving the proportional guarantee of every agent.
Applying this split recursively attains the optimal cut bound when the number of agents is a power of two.

\begin{theorem}\label{thm:moving-power-two}
For every nonatomic valuation profile, if $n$ is a power of two, an equal-length proportional allocation can be produced using at most $n-1$ two-knife queries and at most $2n-2$ cuts.
\end{theorem}

\begin{proof}
We first describe how to split an instance $(C,S)$ with $m=2r$ agents into two equal subinstances using one \textsc{MoveLength} query.
Write $L=\length{C}$ and, for $0\le t\le L/2$, let $A_t=C\langle t,t+L/2\rangle$ and $B_t=C\setminus A_t$.
Thus, $A_t$ and $B_t$ both have length $L/2$.
The two knives move from the two boundaries of $A_0$ to those of $A_{L/2}$, and agent $i$ raises a hand when $v_i(A_t)\ge v_i(C)/2$.
The values $v_i(C)$, and hence the thresholds, can be obtained beforehand by evaluating the components of $C$.
These Robertson--Webb queries are not included in the two-knife query count of the theorem.

If at least $r$ agents weakly prefer each of $A_0$ and $B_0$, assign to $A_0$ every agent who strictly prefers it, add indifferent agents until the group has size $r$, and assign the remaining agents to $B_0$.
Otherwise, orient the motion so that fewer than $r$ hands are raised initially.
The portions between and outside the knives exchange roles at the end of the motion, where more than $r$ hands are therefore raised.
The query stops at the first $t^*$ at which at least $r$ hands are raised.
By continuity, fewer than $r$ agents strictly prefer $A_{t^*}$ to $B_{t^*}$; otherwise, at least $r$ hands would have been raised immediately before $t^*$.
Assign to $A_{t^*}$ any $r$ agents with raised hands, including all agents who strictly prefer it, and assign the remaining agents to $B_{t^*}$.
The first group values $A_{t^*}$ at least $v_i(C)/2$ by the stopping rule, while every agent in the second group weakly prefers $B_{t^*}$ and therefore values it at least $v_i(C)/2$.
This split uses one two-knife query and introduces at most two cut points.

We now recurse on instances $(C,S)$ satisfying $\length{C}=m/n$ and $v_i(C)\ge m/n$ for every $i\in S$, where $m=|S|$.
Initially $C=[0,1]$, $S$ is the full set of agents, and $m=n$.
If $m=1$, assign $C$ to its only agent.
Otherwise $m$ is even, and the split above produces two instances with $m/2$ agents.
Every agent assigned to a child values it at least
$v_i(C)/2\ge m/(2n)$, and each child has length $m/(2n)$.
Thus, the invariant is preserved.

The recursion is a full binary tree with $n$ leaves and hence $n-1$ internal nodes.
Each internal node uses one two-knife query and introduces at most two cut points.
The resulting bounds are therefore at most $n-1$ two-knife queries and $2n-2$ cuts.
Together with \Cref{prop:elementary-lower-bound}, this also proves that the cut bound is optimal.
\end{proof}

\subsection{General Numbers of Agents}

In this subsection, we show how to produce an equal-length proportional allocation for arbitrary $n$.
The procedure uses $O(n^5)$ Robertson--Webb queries, $O(n^2)$ two-knife queries, and $O(n\log n)$ cuts.
The procedure recursively assigns bundles to a nonempty group of agents and continues with the remaining agents and cake.
At each stage, one agent divides the current subcake into one equal-length bundle per remaining agent, with all bundles having the same value to that agent.
A matching argument then assigns a nonempty collection of these bundles to the same number of agents, giving every assigned agent value at least $1/n$.
The lemma below constructs the partition used in this matching argument.
In its application, the equal-value condition makes the chosen agent adjacent to every bundle.
The bound on new cut points yields the total bound of $O(n\log n)$ cuts.
The query bounds follow by summing the cost of the construction over the recursive stages.

\begin{lemma}\label{lem:equal-partition}
Fix an agent $i$, and let $C$ be a finite union of intervals satisfying $\length{C}=k\ell$ and $v_i(C)=k\alpha$ for an integer $k\ge1$ and some $\ell,\alpha>0$.
One can partition $C$ into bundles $T_1,\dots,T_k$ satisfying $\length{T_j}=\ell$ and $v_i(T_j)=\alpha$ for every $j$.
At most $4\lceil\log_2 k\rceil$ new cut points are introduced in obtaining any one bundle.
\end{lemma}

\begin{proof}
We first show how to extract from a current subcake $D$ a subpiece of specified length and value.
Suppose that $D$ is a union of $q_D$ intervals, $\length{D}=t\ell'$, and $v_i(D)=t\alpha'$ for an integer $t\ge2$.
Divide the copy of $D$ constructed above into $t$ consecutive blocks of length $\ell'$ and evaluate them using $O(tq_D)$ queries.
Their values have average $\alpha'$.
If one block has value $\alpha'$, it is the desired subpiece.
Otherwise, there are two consecutive blocks $D\langle a,b\rangle$ and $D\langle b,c\rangle$ whose values lie on opposite sides of $\alpha'$.
Suppose first that $v_i(D\langle a,b\rangle)<\alpha'<v_i(D\langle b,c\rangle)$.
Place the left knife at $a$ and the right knife at a position $y_0$ such that $v_i(D\langle a,y_0\rangle)=\alpha'$.
The first inequality implies $y_0>b$.
For every $x\in[a,b]$, we have $v_i(D\langle x,c\rangle)\ge v_i(D\langle b,c\rangle)>\alpha'$.
Thus, a position of the right knife preserving the value $\alpha'$ exists throughout the motion.
Move the left knife from $a$ to $b$ while agent $i$ moves the right knife continuously and monotonically to keep the intervening value equal to $\alpha'$.
Let $y_1$ be the position of the right knife when the left knife reaches $b$.
The second inequality implies $y_1<c$.
Hence, the length of the intervening portion starts above $\ell'$ and ends below $\ell'$.
If $v_i(D\langle a,b\rangle)>\alpha'>v_i(D\langle b,c\rangle)$, place the right knife at $b$ and the left knife at a position $x_0$ such that $v_i(D\langle x_0,b\rangle)=\alpha'$.
The first inequality implies $x_0>a$.
For every $y\in[b,c]$, we have $v_i(D\langle a,y\rangle)\ge v_i(D\langle a,b\rangle)>\alpha'$.
Thus, a position of the left knife preserving the value $\alpha'$ exists throughout the motion.
Move the right knife from $b$ to $c$ while agent $i$ moves the left knife continuously and monotonically to keep the intervening value equal to $\alpha'$.
Let $x_1$ be the position of the left knife when the right knife reaches $c$.
The second inequality implies $x_1<b$.
Hence, the length starts below $\ell'$ and ends above $\ell'$.
In either case, continuity implies that a \textsc{MoveValue}$_i$ query returns a subpiece of length $\ell'$ and value $\alpha'$.
Thus, the extraction uses $O(tq_D)$ Robertson--Webb queries and at most one two-knife query.

For a current subcake $D$, let $h$ be the integer satisfying $\length{D}=h\ell$ and $v_i(D)=h\alpha$.
If $h=1$, then $D$ is a final bundle.
If $h$ is even, use the extraction above with $t=2$, target length $h\ell/2$, and target value $h\alpha/2$.
The returned subpiece and its complement both have length $h\ell/2$ and value $h\alpha/2$.

If $h\ge3$ is odd, first use the extraction with $t=h$ to obtain one final bundle of length $\ell$ and value $\alpha$.
The residual subcake has length $(h-1)\ell$ and value $(h-1)\alpha$.
Apply the even step to split it into two subcakes, each of length $(h-1)\ell/2$ and value $(h-1)\alpha/2$.
Continue recursively until $h=1$.
Every step replaces an integer $h>1$ by smaller positive integers whose sum is $h$.
Thus, the procedure terminates with the required $k$ bundles.
If $h$ is even, it is halved after one split.
If $h$ is odd, it is reduced to $(h-1)/2$ after two splits.
Thus, obtaining any one bundle requires at most $2\lceil\log_2 k\rceil$ splits.
Each split introduces at most two new cut points, proving the last assertion.
\end{proof}

The construction also gives the complexity bounds used below.
Exactly $k-1$ splits are made, each using at most one \textnormal{\textsc{MoveValue}}$_i$ query.
Consequently, if $C$ is a union of $q$ intervals, each bundle is a union of $O(q+\log k)$ intervals.

Using this partition at each stage gives the following result.

\begin{theorem}\label{thm:rw-upper}
For every nonatomic valuation profile, there is a procedure in the Robertson--Webb model augmented with two-knife queries that produces an equal-length proportional allocation.
It uses $O(n^5)$ Robertson--Webb queries and $O(n^2)$ two-knife queries.
The final allocation has $O(n\log n)$ cuts.
\end{theorem}

\begin{proof}
We recurse on subinstances $(C,S)$, where $S$ is a set of $m$ agents and $C$ is a finite union of intervals, under the invariant $\length{C}=m/n$ and $v_i(C)\ge m/n$ for every $i\in S$.
The initial instance satisfies this invariant.

Choose an agent $a\in S$ and use \textnormal{\textsc{Eval}} queries on the components of $C$ to obtain $V=v_a(C)$.
Apply \Cref{lem:equal-partition} with $k=m$, $\ell=1/n$, and $\alpha=V/m$ to partition $C$ into bundles $T_1,\dots,T_m$ satisfying $\length{T_j}=1/n$ and $v_a(T_j)=V/m$ for every $j$.
Since $V\ge m/n$, agent $a$ values every $T_j$ at least $1/n$.

Form a bipartite graph between $S$ and $\mathcal T=\{T_1,\dots,T_m\}$, joining $i$ to $T_j$ when $v_i(T_j)\ge1/n$.
Because the bundles partition $C$ and $v_i(C)\ge m/n$, every agent has a neighbor.
Agent $a$ is adjacent to every bundle.
The graph is obtained by evaluating every bundle for every agent.

Choose a nonempty $\mathcal T'\subseteq\mathcal T$ inclusion-minimal subject to $|N(\mathcal T')|\le|\mathcal T'|$, and put $S'=N(\mathcal T')$.
Such a set exists because $\mathcal T$ itself qualifies, and agent $a$ belongs to $S'$.
We first show that $|S'|=|\mathcal T'|$.
If $|S'|<|\mathcal T'|$, then $|\mathcal T'|\ge2$ because $a\in S'$, and for any $T_j\in\mathcal T'$ we have $|N(\mathcal T'\setminus\{T_j\})|\le|S'|\le|\mathcal T'|-1$, contradicting the minimality of $\mathcal T'$.
Moreover, every nonempty proper subset $U\subset\mathcal T'$ satisfies $|N(U)|>|U|$ by minimality.
Together with $|N(\mathcal T')|=|\mathcal T'|$, this is Hall's condition for $G[S',\mathcal T']$.
Thus, the graph has a perfect matching.
By the definition of $S'$, no agent outside $S'$ is adjacent to a bundle in $\mathcal T'$.
Assign the bundles in $\mathcal T'$ to the agents in $S'$ according to that perfect matching, and recurse on $S\setminus S'$ and $C\setminus\bigcup_{T_j\in\mathcal T'}T_j$.
Every remaining agent values each removed bundle below $1/n$.
Since $|S'|=|\mathcal T'|$, every remaining agent values the residual subcake at least $(m-|S'|)/n$.
Thus, the invariant is preserved.
Because $S'$ is nonempty, the recursion terminates with a valid allocation.

There are at most $n$ recursive stages.
At a stage with $m$ agents, the construction in the proof of \Cref{lem:equal-partition} performs $m-1$ splits.
Consequently, it uses at most $m-1$ two-knife queries and creates $O(m)$ temporary boundary points.
Over all stages, the procedure uses $O(n^2)$ two-knife queries and creates $O(n^2)$ temporary boundary points.
Every subcake encountered therefore has $O(n^2)$ components.
At one stage, the construction performs at most $n-1$ extractions.
Each extraction considers at most $n$ consecutive blocks of a subcake with $O(n^2)$ components and therefore uses $O(n^3)$ Robertson--Webb queries, as shown in the proof of \Cref{lem:equal-partition}.
Thus, constructing the partition at one stage uses $O(n^4)$ Robertson--Webb queries.
Evaluating its at most $n$ bundles for at most $n$ agents also uses $O(n^4)$ queries, because each bundle has $O(n^2)$ components.
The total over at most $n$ stages is therefore $O(n^5)$ Robertson--Webb queries.

At each stage, keep only the boundaries created while obtaining the bundles assigned at that stage.
Any other boundary separates only bundles left for later stages and is internal to the residual subcake.
Thus, it can be discarded.
Boundaries kept from earlier stages have already been charged to bundles assigned at those stages.
By \Cref{lem:equal-partition}, each assigned bundle requires $O(\log n)$ new cut points, and exactly $n$ bundles are assigned over the whole recursion.
Hence, the final allocation uses $O(n\log n)$ cuts.
\end{proof}
\section{Conclusion}\label{sec:conclusion}

We proved that equal-length proportionality has exact cut complexity $2n-2$ for every $n$.
For the stronger equal-length envy-free problem, the same bound is optimal when $n$ is a prime power, and the general upper bound is $O(n^{1.525})$.

No deterministic protocol in the Robertson--Webb model terminates after finitely many queries on every instance and always returns an exact equal-length proportional allocation.
For every $\varepsilon>0$, however, an equal-length $\varepsilon$-proportional allocation with at most $2n-2$ cuts can be computed using $O(n^2/\varepsilon)$ queries.
A polynomial-time alternative in the real-RAM oracle model uses $O(n^3/\varepsilon)$ queries and $O(n^2/\varepsilon)$ cuts.
In the Robertson--Webb model augmented with two-knife queries, the optimal cut bound is attained for powers of two under arbitrary nonatomic valuations.
For general $n$ and arbitrary nonatomic valuations, the augmented model produces an equal-length proportional allocation using $O(n^5)$ Robertson--Webb queries, $O(n^2)$ two-knife queries, and $O(n\log n)$ cuts.
Two questions remain open: whether equal-length envy-freeness always admits $2n-2$ cuts, and whether a protocol in this augmented model attains $2n-2$ cuts for every $n$.

\section*{Acknowledgement}
This work was partially supported by JST ERATO Grant Number JPMJER2301, JST LOTUS Programme Grant Number JPMJ25152681, JSPS KAKENHI Grant Number JP25K00137, and ANRF MATRICS Grant Number MTR/2021/000474.

\section*{Declaration of competing interest}
The authors declare that they have no known competing financial interests or personal relationships that could have appeared to influence the work reported in this paper.

\section*{Data availability}
No data was used for the research described in the article.
\section*{Declaration of AI assistance}
The authors used OpenAI's ChatGPT and Codex to assist with language editing, manuscript organization, and the presentation of proofs and related work. The authors reviewed and verified all mathematical content and citations. The authors assume responsibility for all content.

\bibliographystyle{abbrvnat} 
\bibliography{references}

@article{Austin1982,
  author  = {A. K. Austin},
  title   = {Sharing a Cake},
  journal = {The Mathematical Gazette},
  volume  = {66},
  number  = {437},
  pages   = {212--215},
  year    = {1982},
  doi     = {10.2307/3616548}
}

@inproceedings{procaccia2017lower,
  title={A lower bound for equitable cake cutting},
  author={Procaccia, Ariel D and Wang, Junxing},
  booktitle={Proceedings of the 2017 ACM Conference on Economics and Computation},
  pages={479--495},
  year={2017},
  doi={10.1145/3033274.3085107}
}

@InProceedings{Kawase2025SimultaneouslyFA,
  author    = {Kawase, Yasushi and Roy, Bodhayan and Sanpui, Mohammad Azharuddin},
  title     = {{Simultaneously Fair Allocation of Indivisible Items Across Multiple Dimensions}},
  booktitle = {45th IARCS Annual Conference on Foundations of Software Technology and Theoretical Computer Science (FSTTCS 2025)},
  pages     = {41:1--41:19},
  year      = {2025},
  doi       = {10.4230/LIPIcs.FSTTCS.2025.41}
}

@article{golz2026fair,
  title={Fair division among couples and small groups},
  author={G{\"o}lz, Paul and Yaghoubizade, Hannane},
  journal={Proceedings of the AAAI Conference on Artificial Intelligence},
  volume={40},
  number={20},
  pages={16963--16970},
  year={2026},
  doi={10.1609/aaai.v40i20.38743}
}

@article{la2026bad,
  title={Bad News for Couples: Tight Lower Bounds for Fair Division of Indivisible Items},
  author={la Tour, Max Dupr{\'e}},
  journal={arXiv preprint arXiv:2601.01012},
  year={2026},
  doi={10.48550/arXiv.2601.01012}
}

@article{cseh2020complexity,
  title={The complexity of cake cutting with unequal shares},
  author={Cseh, {\'A}gnes and Fleiner, Tam{\'a}s},
  journal={ACM Transactions on Algorithms},
  volume={16},
  number={3},
  pages={1--21},
  year={2020},
  doi={10.1145/3380742}
}

@inproceedings{aziz2016discrete,
  title={A discrete and bounded envy-free cake cutting protocol for any number of agents},
  author={Aziz, Haris and Mackenzie, Simon},
  booktitle={2016 IEEE 57th Annual Symposium on Foundations of Computer Science (FOCS)},
  pages={416--427},
  year={2016},
  doi={10.1109/FOCS.2016.52}
}

@article{Pikhurko2000OnEC,
  title={On Envy-Free Cake Division},
  author={Oleg Pikhurko},
  journal={The American Mathematical Monthly},
  year={2000},
  volume={107},
  number={8},
  pages={736--738},
  doi={10.2307/2695471}
}

@article{Su1999RentalHS,
  title={Rental Harmony: Sperner's Lemma in Fair Division},
  author={Francis Edward Su},
  journal={The American Mathematical Monthly},
  year={1999},
  volume={106},
  number={10},
  pages={930--942},
  doi={10.2307/2589747}
}

@inproceedings{Bei2012OptimalPC,
  author    = {Bei, Xiaohui and Chen, Ning and Hua, Xia and Tao, Biaoshuai and Yang, Endong},
  title     = {Optimal Proportional Cake Cutting with Connected Pieces},
  booktitle = {Proceedings of the Twenty-Sixth AAAI Conference on Artificial Intelligence (AAAI 2012)},
  pages     = {1263--1269},
  year      = {2012},
  doi       = {10.1609/aaai.v26i1.8243}
}

@article{Amanatidis2023FairDO,
  title={Fair division of indivisible goods: Recent progress and open questions},
  author={Georgios Amanatidis and Haris Aziz and Georgios Birmpas and Aris Filos-Ratsikas and Bo Li and Herv{\'e} Moulin and Alexandros A. Voudouris and Xiaowei Wu},
  journal={Artificial Intelligence},
  year={2023},
  volume={322},
  pages={103965},
  doi={10.1016/j.artint.2023.103965}
}

@incollection{Bouveret2016FairAO,
  author    = {Bouveret, Sylvain and Chevaleyre, Yann and Maudet, Nicolas},
  title     = {Fair Allocation of Indivisible Goods},
  booktitle = {Handbook of Computational Social Choice},
  editor    = {Brandt, Felix and Conitzer, Vincent and Endriss, Ulle and Lang, J{\'e}r{\^o}me and Procaccia, Ariel D.},
  pages     = {284--310},
  publisher = {Cambridge University Press},
  year      = {2016},
  doi       = {10.1017/CBO9781107446984.013}
}

@article{Steinhaus1949SurLD,
  title={Sur la division pragmatique},
  author={Hugo Steinhaus},
  journal={Econometrica},
  year={1949},
  volume={17},
  pages={315--319},
  doi={10.2307/1907319}
}

@inproceedings{bu2023fair,
  author    = {Bu, Xiaolin and Li, Zihao and Liu, Shengxin and Song, Jiaxin and Tao, Biaoshuai},
  title     = {Fair Division with Allocator's Preference},
  booktitle = {Proceedings of the 19th Conference on Web and Internet Economics (WINE 2023)},
  pages     = {77--94},
  year      = {2023},
  doi       = {10.1007/978-3-031-48974-7_5}
}

@article{stromquistwoodall1985,
  author  = {Stromquist, Walter and Woodall, D. R.},
  title   = {Sets on which several measures agree},
  journal = {Journal of Mathematical Analysis and Applications},
  volume  = {108},
  number  = {1},
  pages   = {241--248},
  year    = {1985},
  doi     = {10.1016/0022-247X(85)90021-6}
}

@InProceedings{goldberg_li_itcs2023,
  author    = {Goldberg, Paul and Li, Jiawei},
  title     = {{Consensus Division in an Arbitrary Ratio}},
  booktitle = {14th Innovations in Theoretical Computer Science Conference (ITCS 2023)},
  pages     = {57:1--57:18},
  year      = {2023},
  doi       = {10.4230/LIPIcs.ITCS.2023.57}
}

@article{jojic_panina_zivaljevic_2021,
  author  = {Joji{\'c}, Du{\v{s}}ko and Panina, Gaiane and {\v{Z}}ivaljevi{\'c}, Rade T.},
  title   = {Splitting Necklaces, with Constraints},
  journal = {SIAM Journal on Discrete Mathematics},
  volume  = {35},
  number  = {2},
  pages   = {1268--1286},
  year    = {2021},
  doi     = {10.1137/20M1331949}
}

@article{segal_halevi_2019,
  title={Cake-cutting with different entitlements: How many cuts are needed?},
  author={Segal-Halevi, Erel},
  journal={Journal of Mathematical Analysis and Applications},
  volume={480},
  number={1},
  pages={123382},
  year={2019},
  doi={10.1016/j.jmaa.2019.123382}
}

@article{crew_narayanan_spirkl_2020,
  author  = {Crew, Tom and Narayanan, Bhargav and Spirkl, Sophie},
  title   = {Disproportionate Division},
  journal = {Bulletin of the London Mathematical Society},
  volume  = {52},
  number  = {5},
  pages   = {885--890},
  year    = {2020},
  doi     = {10.1112/blms.12368}
}

@article{barman_ebadian_latifian_shah_2025,
  author  = {Barman, Siddharth and Ebadian, Soroush and Latifian, Mohamad and Shah, Nisarg},
  title   = {Fair Division with Market Values},
  journal = {Proceedings of the AAAI Conference on Artificial Intelligence},
  volume  = {39},
  number  = {13},
  pages   = {13589--13596},
  year    = {2025},
  doi     = {10.1609/aaai.v39i13.33484}
}

@inproceedings{biswas_barman_2018,
  author    = {Biswas, Arpita and Barman, Siddharth},
  title     = {Fair Division Under Cardinality Constraints},
  booktitle = {Proceedings of the Twenty-Seventh International Joint Conference on Artificial Intelligence (IJCAI 2018)},
  pages     = {91--97},
  year      = {2018},
  doi       = {10.24963/ijcai.2018/13}
}

@article{kawase_mahara_2026,
  title={Fair and Efficient Balanced Allocation for Indivisible Goods},
  author={Kawase, Yasushi and Mahara, Ryoga},
  journal={Proceedings of the AAAI Conference on Artificial Intelligence},
  volume={40},
  number={20},
  pages={17067--17075},
  year={2026},
  doi={10.1609/aaai.v40i20.38755}
}

@article{baker_harman_pintz_2001,
  author  = {Baker, R. C. and Harman, G. and Pintz, J.},
  title   = {The Difference Between Consecutive Primes, {II}},
  journal = {Proceedings of the London Mathematical Society},
  volume  = {83},
  number  = {3},
  pages   = {532--562},
  year    = {2001},
  doi     = {10.1112/plms/83.3.532}
}

@article{dudek_2015,
  author  = {Dudek, Adrian W.},
  title   = {On the {Riemann} Hypothesis and the Difference Between Primes},
  journal = {International Journal of Number Theory},
  volume  = {11},
  number  = {3},
  pages   = {771--778},
  year    = {2015},
  doi     = {10.1142/S1793042115500426}
}

@article{cramer_1936,
  author  = {Cram{\'e}r, Harald},
  title   = {On the Order of Magnitude of the Difference Between Consecutive Prime Numbers},
  journal = {Acta Arithmetica},
  volume  = {2},
  pages   = {23--46},
  year    = {1936}
}

@article{alon_1987,
  author  = {Alon, Noga},
  title   = {Splitting Necklaces},
  journal = {Advances in Mathematics},
  volume  = {63},
  number  = {3},
  pages   = {247--253},
  year    = {1987},
  doi     = {10.1016/0001-8708(87)90055-7}
}

@article{robertson_webb_1997,
  author  = {Robertson, Jack M. and Webb, William A.},
  title   = {Near Exact and Envy Free Cake Division},
  journal = {Ars Combinatoria},
  volume  = {45},
  pages   = {97--108},
  year    = {1997}
}

@article{edmonds_pruhs_2011,
  author  = {Edmonds, Jeff and Pruhs, Kirk},
  title   = {Cake Cutting Really Is Not a Piece of Cake},
  journal = {ACM Transactions on Algorithms},
  volume  = {7},
  number  = {4},
  pages   = {51:1--51:12},
  year    = {2011},
  doi     = {10.1145/2000807.2000819}
}

@misc{arndt_moseley_im_pruhs_2026,
  author        = {Arndt, Stephen and Moseley, Benjamin and Im, Sungjin and Pruhs, Kirk},
  title         = {No, Cake Cutting Really Is a Piece of Cake},
  year          = {2026},
  eprint        = {2606.07238},
  archiveprefix = {arXiv},
  primaryclass  = {cs.DS},
  note          = {arXiv:2606.07238}
}

@inproceedings{kawase_sanpui_2026,
  author    = {Kawase, Yasushi and Sanpui, Mohammad Azharuddin},
  title     = {Exact Cut Complexity of Equal-Length Proportional Cake Cutting},
  booktitle = {51st International Symposium on Mathematical Foundations of Computer Science (MFCS 2026)},
  series    = {Leibniz International Proceedings in Informatics (LIPIcs)},
  volume    = {386},
  pages     = {34:1--34:16},
  publisher = {Schloss Dagstuhl--Leibniz-Zentrum f{\"u}r Informatik},
  year      = {2026},
  doi       = {10.4230/LIPIcs.MFCS.2026.34}
}

\end{document}